\documentclass[10pt]{article}

\usepackage{graphicx}%
\usepackage{multirow}%
\usepackage{amsmath,amssymb,amsfonts}%
\usepackage{amsthm}%
\usepackage{mathrsfs}%
\usepackage[title]{appendix}%
\usepackage{xcolor}%
\usepackage{textcomp}%
\usepackage{manyfoot}%
\usepackage{booktabs}%
\usepackage{geometry}
\usepackage{algpseudocode}%
\usepackage{listings}%
\usepackage{cleveref}

\newtheorem{theorem}{Theorem}
\newtheorem{lemma}{Lemma}
\newtheorem{corollary}{Corollary}
\newtheorem{remark}{Remark}%
\newtheorem{conjecture}{Conjecture}%

\crefname{equation}{Eq.}{Eq.}
\crefname{theorem}{Theorem}{Theorem}
\crefname{corollary}{Corollary}{Corollary}
\crefname{lemma}{Lemma}{Lemma}
\crefname{remark}{Remark}{Remark}
\crefname{definition}{Definition}{Definition}
\crefname{example}{Example}{Example}
\crefname{section}{Section}{Section}

\begin{document}
	
	\begin{center}
		
		{\large  \bf Infinite families of MDS and almost MDS codes from BCH codes\footnotemark[2]}
		\footnotetext[2]{Supported by NSFC (Nos. 11971102, 12371035, 11801070, 12171241)}
		\vskip 0.8cm
		{\small Haojie Xu $^1$, Xia Wu$^{1*}$, Wei Lu$^1$, Xiwang Cao$^2$}\\
		
		{\small $^1$School of Mathematics, Southeast University, Nanjing 210096, China}\\
		{\small $^2$Department of Math, Nanjing University of Aeronautics and Astronautics, Nanjing 211100, China}\\
		{\small E-mail:
			xuhaojiechn@163.com, wuxia80@seu.edu.cn, luwei1010@seu.edu.cn, xwcao@nuaa.edu.cn}\\
		{\small $^*$Corresponding author. (Email: luwei1010@seu.edu.cn)}
		\vskip 0.8cm
	\end{center}

	{\bf Abstract:}
	In this paper, the sufficient and necessary condition for the minimum distance of the BCH codes over $\mathbb{F}_q$ with length $q+1$ and designed distance 3 to be 3 and 4 are provided. Let $d$ be the minimum distance of the BCH code $\mathcal{C}_{(q,q+1,3,h)}$. We prove that (1) for any $q$, $d=3$ if and only if $\gcd(2h+1,q+1)>1$; (2) for $q$ odd, $d=4$ if and only if $\gcd(2h+1,q+1)=1$. By combining these conditions with the dimensions of these codes, the parameters of this BCH code are determined completely when $q$ is odd. Moreover, several infinite families of MDS and almost MDS (AMDS) codes are shown. Furthermore, the sufficient conditions for these AMDS codes to be distance-optimal and dimension-optimal locally repairable codes are presented. Based on these conditions, several examples are also given.

	

	\section{Introduction}\label{sec1}
	
	Let $q$ be the power of a prime $p$ and $\mathbb{F}_q$ be the finite field with $q$ elements. An $[n,k]$ linear code $\mathcal{C}$ over $\mathbb{F}_q$ is a vector space of $\mathbb{F}_q^n$ with dimension $k$. The minimum distance $d$ of a linear code $\mathcal{C}$ is represented by $d=\min\{\mathrm{wt}(\mathbf{c}):\mathbf{c}\in\mathcal{C}\setminus \{\mathbf{0}\}\}$, where $\mathrm{wt}(\mathbf{c})$ is the number of nonzero coordinates of $\mathbf{c}$ and called the Hamming weight of $\mathbf{c}$. The dual code $\mathcal{C}^\perp$ of $\mathcal{C}$ is defined as 
	\begin{equation*}
		\mathcal{C}^\perp=\left\{\mathbf{v}\in\mathbb{F}_q^n:\mathbf{c}\cdot\mathbf{v}=0,\ \forall\;\mathbf{c}\in\mathcal{C}\right\},
	\end{equation*}
	where $\mathbf{c}\cdot\mathbf{v}$ is the inner product of $\mathbf{c}$ and $\mathbf{v}$. It is desired to design linear codes with the largest possible rate $\frac{k}{n}$ and minimum distance $d$ in coding theory. Nevertheless, there are some tradeoffs among $n$, $k$ and $d$. The Singleton bound indicates that $d\le n-k+1$ for a linear  code with parameters $[n,k,d]$. Linear codes that meet the Singleton bound are called maximum distance separable, or MDS for short. Note that if $\mathcal{C}$ is MDS so is the dual code $\mathcal{C}^\perp$. An $[n,k,n-k]$ code is said to be almost MDS (AMDS for short) \cite{Boer1996AlmostMDSCodes}. Unlike MDS code, the dual code of an AMDS code need not be AMDS. Furthermore, $\mathcal{C}$ is said to be near MDS (NMDS for short) if both $\mathcal{C}$ and its dual code $\mathcal{C}^\perp$ are AMDS codes \cite{Dodunekov1995NearMDSCodes}. MDS, NMDS and AMDS codes play vital roles in communications, data storage, combinatorial theory, and secret sharing \cite{Ding2020InfiniteFamiliesMDS,Tang2021InfiniteFamilyLinear,Heng2023NewInfiniteFamilies,Dodunekova1997AlmostMDSNearMDSCodes,Zhou2009SecretSharingScheme}.
	
	BCH codes are a crucial class of linear codes due to its exceptional error-correcting capabilities and simple encoding and decoding algorithms. Ding and Tang demonstrated that the narrow-sense BCH codes $\mathcal{C}_{(2^s,2^s+1,3,1)}$ for $s$ even, $\mathcal{C}_{(2^s,2^s+1,4,1)}$ for $s$ odd, and $\mathcal{C}_{(3^s,3^s+1,3,1)}$ are all NMDS codes \cite{Ding2020InfiniteFamiliesMDS,Tang2021InfiniteFamilyLinear}. Moreover, it was presented that the above NMDS codes constructed by Ding and Tang are $d$-optimal and $k$-optimal locally repairable codes by Tan et al. in \cite{Tan2023MinimumLocalityLinear}. Geng et al. showed that the BCH codes $\mathcal{C}_{(3^s,3^s+1,3,4)}$ with $s$ being odd is an AMDS code, which is also distance-optimal ($d$-optimal for short) and dimension-optimal ($k$-optimal for short) locally repair code \cite{Geng2022ClassAlmostMDS}. Locally repairable codes are widely used in distributed and cloud storage systems. There are some other works \cite{Li2023ConstructionsMDSCodes,Li2023ConstructionOptimalLocally,Heng2023MDSCodesDimension,Fu2023ConstructionSingletonTypeOptimal} related to constructing AMDS or NMDS codes to obtain locally repairable codes with optimal dimensions or distances in recent years. 
	Motivated by the tasks of \cite{Ding2020InfiniteFamiliesMDS,Geng2022ClassAlmostMDS}, the objective of this paper is to study the parameters of the BCH code $\mathcal{C}_{(q,q+1,3,h)}$ and its dual code. We establish the sufficient and necessary conditions for the minimum distance $d$ of $\mathcal{C}_{(q,q+1,3,h)}$ to be 3 and 4, respectively. Applying these conditions, several infinite classes of MDS codes and AMDS codes are provided. Furthermore, we determine that these AMDS codes are $d$-optimal and $k$-optimal locally repairable codes when $q>4h$.
	
	The rest of this paper is organized as follows. In \cref{sec2}, we provide some concepts of cyclic codes, BCH codes, and locally repairable codes. In \cref{sec3}, we present the sufficient and necessary conditions for the minimum distance $d$ of $\mathcal{C}_{(q,q+1,3,h)}$ to be 3 and 4, respectively. In addition, we obtain several infinite families of MDS codes and AMDS codes. In \cref{sec4}, we give several classes of AMDS codes which are $d$-optimal and $k$-optimal locally repairable codes. In \cref{sec5}, we conclude this paper.
	
	\section{Preliminaries}\label{sec2}
	
	\subsection{Cyclic codes and BCH codes}\label{sec2.1}
	
	An $[n,k]$ code $\mathcal{C}$ over $\mathbb{F}_q$ is called \textit{cyclic} if $\mathbf{c}=(c_0,c_1,\cdots,c_{n-2},c_{n-1})\in\mathcal{C}$ implies $(c_{n-1},c_0,c_1,\cdots,c_{n-2})\in\mathcal{C}$. Throughout this subsection, we assume that $\gcd(n,q)=1$. The residue class ring $\mathbb{F}_q[x]/(x^n-1)$ is isomorphic to $\mathbb{F}_q^n$ as a vector space over $\mathbb{F}_q$, where the isomorphism is given by $(c_0,c_1,\cdots,c_{n-1})\leftrightarrow c_0+c_1x+\cdots+c_{n-1}x^{n-1}$. Because of this isomorphism, any linear code $\mathcal{C}$ of length $n$ over $\mathbb{F}_q$ corresponds to a subset of $\mathbb{F}_q[x]/(x^n-1)$. Then we can interpret $\mathcal{C}$ as a subset of $\mathbb{F}_q[x]/(x^n-1)$. It is easily checked that the linear code $\mathcal{C}$ is cyclic if and only if $\mathcal{C}$ is an ideal of $\mathbb{F}_q[x]/(x^n-1)$.
	
	Since every ideal of $\mathbb{F}_q[x]/(x^n-1)$ is principal, every cyclic code $\mathcal{C}$ is generated by the unique monic polynomial $g(x)\in\mathbb{F}_q[x]$ of the lowest degree, i.e., $\mathcal{C}=\left<g(x)\right>$. $g(x)$ is called the \textit{generator polynomial} of $\mathcal{C}$ and $h(x)=(x^n-1)/g(x)$ is called the \textit{parity-check polynomial} of $\mathcal{C}$. It should be noticed that the generator polynomial $g(x)$ is a factor of $x^n-1$, thus we have to study the canonical factorization of $x^n-1$ over $\mathbb{F}_q$ to handle $g(x)$. With this intention, we need to introduce $q$-cyclotomic cosets modulo $n$ \cite{Ding2014CodesDifferenceSetsa}.
	
	Let $\mathbb{Z}_n$ denotes the set $\{0,1,2,\cdots,n-1\}$ and let $s<n$ be a nonnegative integer. The \textit{$q$-cyclotomic coset of s modulo $n$} is given by 
	\begin{equation}\label{sec2.1 equ1}
		C_s=\{s,sq,sq^2,\ldots,sq^{\ell_s-1}\}\bmod n\subseteq\mathbb{Z}_n,
	\end{equation}
	where $\ell_s$ is the smallest positive integer such that $s \equiv sq^{\ell_s} \pmod n$, and is the size of the $q$-cyclotomic coset. The smallest integer in $C_s$ is called the \textit{coset leader} of $C_s$. Let $\Gamma_{(n,q)}$ be the set of all the coset leaders. We have then $C_s \cap C_t = \emptyset$ for any two distinct elements $s$ and $t$ in $\Gamma_{(n,q)}$, and
	\begin{equation}\label{sec2.1 equ2}
		\bigcup_{s\in\Gamma_{(n,q)}}C_s=\mathbb{Z}_n.
	\end{equation}
	Therefore, the distinct $q$-cyclotomic cosets modulo $n$ partition $\mathbb{Z}_n$. 
		
	Let $m$ be the multiplicative order of $q$ modulo $n$. Let $\alpha$ be a generator of $\mathbb{F}_{q^m}$ and $\beta=\alpha^{(q^m-1)/n}$. Then $\beta$ is a primitive $n$-th root of unity in $\mathbb{F}_{q^m}$. Denote the minimal	polynomial over $\mathbb{F}_q$ of $\beta^s$ by $\mathrm{M}_{\beta^s}(x)$. And this polynomial is given by
	\begin{equation}\label{sec2.1 equ3}
		\mathrm{M}_{\beta^s}(x)=\prod\limits_{i\in C_s}(x-\beta^i)\in\mathbb{F}_q[x],
	\end{equation}
	which is irreducible over $\mathbb{F}_q$. It then follows from \cref{sec2.1 equ1} that
	\begin{equation*}
		x^n-1=\prod_{s\in\Gamma_{(n,q)}}\mathrm{M}_{\beta^s}(x),
	\end{equation*}
	which is the canonical factorization of $x^n - 1$ over $\mathbb{F}_q$.
	
	We proceed to recall a crucial family of cyclic codes, namely BCH codes. Let $h$ be a nonnegative integer and $m$ be the multiplicative order of $q$ modulo $n$. Suppose that $\beta\in\mathbb{F}_{q^m}$ is a primitive $n$-th root of unity and $\delta$ is an integer with $2\le\delta\le n$. Then a \textit{BCH code} over $\mathbb{F}_q$ of length $n$ and \textit{designed distance} $\delta$, denoted by $\mathcal{C}_{(q,n,\delta,h)}$, is a cyclic code generated by 
	\begin{equation*}
		g_{(q,n,\delta,h)}=\mathrm{lcm}\left(\mathrm{M}_{\beta^h}(x),\mathrm{M}_{\beta^{h+1}}(x),\cdots,\mathrm{M}_{\beta^{h+\delta-2}}(x)\right),
	\end{equation*}
	where the least common multiple is computed over $\mathbb{F}_q$ and $\mathrm{M}_{\beta^s}(x)$ is represented in \cref{sec2.1 equ3}. If $h=1$, the BCH code $\mathcal{C}_{(q,n,\delta,h)}$ is called \textit{narrow-sense BCH code}. If $n=q^m-1$, the corresponding BCH codes are called \textit{primitive}.  According to BCH bound, the minimum distance of the BCH code $\mathcal{C}_{(q,n,\delta,h)}$ is at least $\delta$.
	
	\subsection{Locally repairable codes}
	
	A \textit{locally repairable code} (LRC) with \textit{locality} $r$ is a code that allows the repair of any symbol within a codeword by accessing at most $r$ other symbols in the codeword. A linear code of code length	$n$, dimension $k$, minimum distance $d$ and locality $r$ over $\mathbb{F}_q$ is referred to as an $(n, k, d, q; r)$-LRC \cite{Tan2019OptimalCyclicLocally}. Similar to Singleton bound, there are some tradeoffs among the locality, length, dimension,  and minimum distance of LRCs.
	
	\begin{lemma}[\textit{Singleton-like bound} \cite{Gopalan2012LocalityCodewordSymbols}]
		\label{sec2 lem1}
		For an $(n, k, d, q; r)$-LRC, we have
		\begin{equation}\label{sec2.2 equ1}
			d\leq n-k-\left\lceil\frac{k}{r}\right\rceil+2.
		\end{equation}
	\end{lemma}
	
	\begin{lemma}[\textit{Cadambe-Mazumdar bound} \cite{Cadambe2015BoundsSizeLocally}]
		\label{sec2 lem2}
		For an $(n, k, d, q; r)$-LRC, we have 
		\begin{equation}\label{sec2.2 equ2}
			k\le\min_{t\in\mathbb{N}}\left\{tr+k_{opt}^{(q)}(n-t(r+1),d)\right\},
		\end{equation}
		where $k_{opt}^{(q)}(n,d)$ denotes the largest possible dimension of a linear code with code length $n$, minimum distance $d$ and alphabet size $q$, and $\mathbb{N}$ is the set of all non-negative integers.
	\end{lemma}
	
	For any $(n, k, d, q; r)$-LRC, the code is addressed as distance-optimal if the distance $d$ achieves the Singleton-like bound in \cref{sec2.2 equ1} and dimension-optimal if the dimension $k$ meet the Cadambe-Mazumdar bound in \cref{sec2.2 equ2}.
	
	The locality of a nontrivial cyclic code can be determined by the minimum distance of its dual code \cite{Tan2023MinimumLocalityLinear}.
	\begin{lemma}
		\label{sec2 lem3}
		A nontrivial cyclic code has locality $d(\mathcal{C}^\perp)-1$.
	\end{lemma}
	
	\section{Infinite families of MDS and almost MDS codes}\label{sec3}
	Throughout this section, let $q$ be the power of a prime $p$, $0\le h\le q$ be a positive integer, and $U_{q+1}$ be the set of all $(q+1)$-th roots of unity in $\mathbb{F}_{q^2}$. In this section, we consider the BCH code $\mathcal{C}_{(q,q+1,3,h)}$ and its dual. 
	
	We will first examine the dimensions of $\mathcal{C}_{(q,q+1,3,h)}$.
	\begin{theorem}\label{sec3 thm0}
		Suppose that the dimensions of $\mathcal{C}_{(q,q+1,3,h)}$ is $k$. Then for $q$ even,
		\begin{equation*}
			k=\left\{\begin{array}{ll}
				q-1\quad & \mathrm{if}\ h=\frac{q}{2},\\
				q-2\quad & \mathrm{if}\ h=0\ \mathrm{or}\ q,\\
				q-3\quad & \mathrm{if}\ \mathrm{otherwise};
			\end{array}\right.
		\end{equation*}
		for $q$ odd,
		\begin{equation*}
			k=\left\{\begin{array}{ll}
				q-2\quad & \mathrm{if}\ h=0,\frac{q-1}{2},\frac{q+1}{2}\ \mathrm{or}\ q,\\
				q-3\quad & \mathrm{if}\ \mathrm{otherwise}.
			\end{array}\right.
		\end{equation*}
	\end{theorem}
	\begin{proof}
		Let $\alpha$ be a generator of $\mathbb{F}_{q^2}^{\ast}$ and $\beta=\alpha^{q-1}$. It is clear that $\beta$ is a primitive $(q+1)$-th root of unity in $\mathbb{F}_{q^2}$. Let $g_h(x)$ and $g_{h+1}(x)$ denote the minimal polynomials of $\beta^h$ and $\beta^{h+1}$ over $\mathbb{F}_q$, respectively. By definition, $g(x)=\mathrm{lcm}\left(g_h(x),g_{h+1}(x)\right)$ is the generator polynomial of $\mathcal{C}_{(q,q+1,3,h)}$. Therefore, the dimension of $\mathcal{C}_{(q,q+1,3,h)}$ is $q+1-\deg(g(x))$. That is to say, we only need to compute the degree of $g(x)$.
		
		Both $g_h(x)$ and $g_{h+1}(x)$ can be calculated from \cref{sec2.1 equ3}. Thus we first present the $q$-cyclotomic cosets modulo $q+1$ according to \cref{sec2.1 equ1}. For $q$ even, the $q$-cyclotomic cosets modulo $q+1$ are 
		\begin{equation}\label{sec3 thm0 equ1}
			\{0\},\{1,q\},\ldots,\{i,q+1-i\},\ldots,\{q/2,(q+2)/2\}.
		\end{equation}
		For $q$ odd, the $q$-cyclotomic cosets modulo $n$ are 
		\begin{equation}\label{sec3 thm0 equ2}
			\{0\},\{1,q\},\ldots,\{i,q+1-i\},\ldots,\{(q-1)/2,(q+3)/2\},\{(q+1)/2\}.
		\end{equation}
		For any $q$, if $h=0\ \mathrm{or}\ q$, we have $g(x)=(x-\beta^0)(x-\beta^1)(x-\beta^q)$ by \cref{sec3 thm0 equ1,sec3 thm0 equ2}. Since $\beta$ is a primitive $(q+1)$-th root of unity, $\deg(g(x))=3$.
		
		When $q$ is even, $\frac{q}{2}$ is an integer. 
		\begin{itemize}
			\item If $h=\frac{q}{2}$, then $h+1=\frac{q}{2}+1$, and so $h+(h+1)=q+1$. It follows from \cref{sec3 thm0 equ1} that $g_h(x)=g_{h+1}(x)$. Hence $g(x)=g_h(x)$. Then $\deg(g(x))=2$.
			\item If $h\notin\{0,\frac{q}{2},q\}$, then $h,h+1,q-h,q+1-h$ are pairwise distinct. It follows from \cref{sec3 thm0 equ1} that $g_h(x)$ has only roots $\beta^{h}$ and $\beta^{q+1-h}$ and $g_{h+1}(x)$ has roots $\beta^{h+1}$ and $\beta^{q-h}$. Thus one can derive that $g_h(x)$ and $g_{h+1}(x)$ are distinct irreducible polynomials of degree 2. Then $\deg(g(x))=4$.
		\end{itemize}
		
		When $q$ is odd, $\frac{q-1}{2}$ and $\frac{q+1}{2}$ are both integers. 
		\begin{itemize}
			\item If $h=\frac{q-1}{2}\ \mathrm{or}\ \frac{q+1}{2}$, we have $g(x)=(x-\beta^{\frac{q-1}{2}})(x-\beta^{\frac{q+1}{2}})(x-\beta^\frac{q+3}{2})$ by \cref{sec3 thm0 equ2}. Since $\beta$ is a primitive $(q+1)$-th root of unity, $\deg(g(x))=3$.
			\item If $h\notin\{0,\frac{q-1}{2},\frac{q+1}{2},q\}$, then $h,h+1,q-h,q+1-h$ are pairwise distinct. Thus $g_h(x)$ and $g_{h+1}(x)$ are distinct irreducible polynomials of degree 2. Then $\deg(g(x))=4$.
		\end{itemize}
		This completes the proof.  
	\end{proof}
	 
	The dimensions of $\mathcal{C}_{(q,q+1,3,h)}^\perp$ can be directly obtained from \cref{sec3 thm0}.
	\begin{corollary}\label{sec3 cor0}
		Suppose that the dimensions of $\mathcal{C}_{(q,q+1,3,h)}^\perp$ is $k^\perp$. Then for $q$ even,
		\begin{equation*}
			k^\perp=\left\{\begin{array}{ll}
				2\quad & \mathrm{if}\ h=\frac{q}{2},\\
				3\quad & \mathrm{if}\ h=0\ \mathrm{or}\ q,\\
				4\quad & \mathrm{if}\ \mathrm{otherwise};
			\end{array}\right.
		\end{equation*}
		for $q$ odd,
		\begin{equation*}
			k^\perp=\left\{\begin{array}{ll}
				3\quad & \mathrm{if}\ h=0,\frac{q-1}{2},\frac{q+1}{2}\ \mathrm{or}\ q,\\
				4\quad & \mathrm{if}\ \mathrm{otherwise}.
			\end{array}\right.
		\end{equation*}
	\end{corollary}

	We will next consider the minimum distance of the BCH code $\mathcal{C}_{(q,q+1,3,h)}$. To begin with, we display a necessary lemma whose proof is straightforward.
	\begin{lemma}
		\label{sec3 lem1}
		Let $x, y, z \in \mathbb{F}_{q^2}^{\ast}$. Then
		\begin{equation*}
			\begin{aligned}&\left|\begin{array}{lll}
					x^{-h} & y^{-h} & z^{-h}\\ x^h & y^h & z^h\\ x^{h+1} & y^{h+1} & z^{h+1}
				\end{array}\right|\\
				=&x^{-h}y^{-h}z^{-h}\left[\left(x^{2h+1}-y^{2h+1}\right)\left(z^{2h}-y^{2h}\right)-\left(z^{2h+1}-y^{2h+1}\right)\left(x^{2h}-y^{2h}\right)\right],
			\end{aligned}
		\end{equation*}
		and 
		\begin{equation*}
			\begin{aligned}
				&\left|\begin{array}{lll}x^{-(h+1)} & y^{-(h+1)} & z^{-(h+1)}\\ x^h & y^h & z^h\\ x^{h+1} & y^{h+1} & z^{h+1}\end{array}\right|\\
				=&x^{-(h+1)}y^{-(h+1)}z^{-(h+1)}\cdot\\
				&\left[\left(x^{2h+2}-y^{2h+2}\right)\left(z^{2h+1}-y^{2h+1}\right)-\left(z^{2h+2}-y^{2h+2}\right)\left(x^{2h+1}-y^{2h+1}\right)\right].
			\end{aligned}
		\end{equation*}
	\end{lemma}
	
	\begin{lemma}\label{sec3 lem2}
		Assume that $x,y,z$ are three pairwise distinct elements in $U_{q+1}$. If
		\begin{equation*}
			\left|\begin{array}{lll}
				x^{-h} & y^{-h} & z^{-h}\\ x^h & y^h & z^h\\ x^{h+1} & y^{h+1} & z^{h+1}
			\end{array}\right|
			= 0,
		\end{equation*}
		then either $x^{2h+1} = y^{2h+1} = z^{2h+1}$ or $x^{2h} = y^{2h} = z^{2h}$.
	\end{lemma}
	\begin{proof}
		It follows from \cref{sec3 lem1} that
		\begin{equation}\label{lem2 equ1}
			\left(x^{2h+1}-y^{2h+1}\right)\left(z^{2h}-y^{2h}\right)=\left(z^{2h+1}-y^{2h+1}\right)\left(x^{2h}-y^{2h}\right).
		\end{equation}
		Raising both sides of \cref{lem2 equ1} to $q$-th power yields
		\begin{equation*}
			\left(x^{-2h+1}-y^{-2h+1}\right)\left(z^{-2h}-y^{-2h}\right)=\left(z^{-2h+1}-y^{-2h+1}\right)\left(x^{-2h}-y^{-2h}\right),
		\end{equation*}
		which is the same as
		\begin{equation}\label{lem2 equ2}
			zy\left(y^{2h+1}-x^{2h+1}\right)\left(y^{2h}-x^{2h}\right)=xy\left(y^{2h+1}-z^{2h+1}\right)\left(y^{2h}-z^{2h}\right).
		\end{equation}
		If neither side of the \cref{lem2 equ1} is equal to 0, then we obtain $x=z$ by combining \cref{lem2 equ1,lem2 equ2}, a contradiction to $x\ne z$. Hence, both sides of the \cref{lem2 equ1} is zero. Since $\gcd(2h,2h+1)=1$, $\frac{x}{y}\ne1$, we have either $x^{2h+1} - y^{2h+1} = 0$ or $x^{2h}-y^{2h} = 0$. Similarly, we have either $z^{2h+1} - y^{2h+1} = 0$ or $z^{2h}-y^{2h} = 0$. Then it follows from \cref{lem2 equ1} that either $x^{2h+1} - y^{2h+1} = 0$ and $z^{2h+1} - y^{2h+1} = 0$, or $x^{2h}-y^{2h} = 0$ and $z^{2h}-y^{2h} = 0$. Therefore, either $x^{2h+1} = y^{2h+1} = z^{2h+1}$ or $x^{2h} = y^{2h} = z^{2h}$.
	\end{proof}
	
	Similarly, for the second determinant stated in \cref{sec3 lem1}, we present the following lemma. Its proof is omitted here as it is similar to the proof of \cref{sec3 lem2}.
	\begin{lemma}
		\label{sec3 lem3}
		Assume that $x,y,z$ are three pairwise distinct elements in $U_{q+1}$. If 
		\begin{equation*}
			\left|\begin{array}{lll}x^{-(h+1)} & y^{-(h+1)} & z^{-(h+1)}\\ x^h & y^h & z^h\\ x^{h+1} & y^{h+1} & z^{h+1}\end{array}\right|
			= 0,
		\end{equation*}
		then either $x^{2h+1} = y^{2h+1} = z^{2h+1}$ or $x^{2h+2} = y^{2h+2} = z^{2h+2}$.
	\end{lemma}
	
	We will also require the following lemma.
	\begin{lemma}
		\label{sec3 lem4}
		Suppose that $x,y,z$ are three pairwise distinct elements in $U_{q+1}$ and $k>2$ is an integer. If $x^{k}=y^{k}=z^{k}$, then $\gcd(k,q+1)>2$.
	\end{lemma}
	\begin{proof}
		If $x^{k}=y^{k}=z^{k}$, then we have $\left(\frac{x}{y}\right)^k=\left(\frac{z}{y}\right)^k=1$. Since $x,y,z$ are pairwise distinct, $\frac{x}{y},\frac{z}{y}$ and 1 are all different from each other. Hence the polynomial $t^k-1\in\mathbb{F}_{q^2}[t]$ have at least three roots in $U_{q+1}$. That is, $\gcd(k,q+1)\ge3$.
		This completes the proof.
	\end{proof}
	
	Denote the minimum distance of $\mathcal{C}_{(q,q+1,3,h)}$ by $d$. We are now ready to prove the following result about the sufficient and necessary condition for $d=3$.
	\begin{theorem}\label{sec3 thm1}
		The minimum distance $d$ of $\mathcal{C}_{(q,q+1,3,h)}$ is equal to 3 if and only if $\gcd(2h+1,q+1)>1$. Equivalently, $4 \le d \le 5$ if and only if $\gcd(2h+1,q+1)=1$.
	\end{theorem}
	\begin{proof}
		Let $\alpha$ be a generator of $\mathbb{F}_{q^2}^{\ast}$ and $\beta=\alpha^{q-1}$. Then $\beta$ is a primitive $(q+1)$-th root of unity in $\mathbb{F}_{q^2}$. Define
		\begin{equation*}
			\left.H=\left[\begin{array}{lllll}1&\beta^{h}&(\beta^{h})^2&\cdots&(\beta^{h})^q\\1&\beta^{h+1}&(\beta^{h+1})^2&\cdots&(\beta^{h+1})^q\end{array}\right.\right].
		\end{equation*}
		It is easily seen that $H$ is a parity-check matrix of $\mathcal{C}_{(q,q+1,3,h)}$.
		
		Assume that $\gcd(2h+1,q+1)>1$, then $\gcd(2h+1,q+1)\ge3$ as $2h+1$ is odd. Hence the polynomial $t^{2h+1}-1\in\mathbb{F}_{q^2}[t]$ have at least three roots in $U_{q+1}$. Denote the set of all such roots by $V$. Suppose that $x\in V\setminus\{1\}$. Next we will show that $x^{h+1}\in V\setminus\{1,x\}$. It is obvious that $x^{h+1}\in V$. We only need to prove $x^{h+1}\ne1$ and $x^{h+1}\ne x$. On the contrary, assume that $x^{h+1}=1$. Then we have $x^{h+1}=x^{2h+1}=1$. Since $\gcd(h+1,2h+1)=1$, thus $x=1$, a contradiction to $x\in V\setminus\{1\}$. If $x^{h+1}= x$, then we have $x^{h}=x^{2h+1}=1$. Since $\gcd(h,2h+1)=1$,  the only possible value for $x$ is 1, which also gives a contradiction. Therefore, $1,x,x^{h+1}$ are three pairwise distinct elements in $V$. Without loss of generality, assume that $x=\beta^{l_1}$ and $x^{h+1}=\beta^{l_2}$, where $0<l_1,l_2\le q$. Define a vector $\mathbf{c}=(c_0,c_1,...c_q)$, where
		\begin{equation*}
			\left.c_i=\left\{\begin{array}{ll}1&\text{if}~i\in\left\{0,l_1\right\},\\-\frac{1+x^h}{x^{h(h+1)}}&\text{if}~i=l_2,\\0&\text{otherwise}.\end{array}\right.\right.
		\end{equation*}
		Note that 
		\begin{equation*}
			\left(\frac{1+x^h}{x^{h(h+1)}}\right)^q=\frac{1+x^{-h}}{(x^{h+1})^{-h}}=\frac{(x^{h+1})^{h}(1+x^{h})}{x^h}=\frac{1+x^{h}}{(x^{-h})^{h}}=\frac{1+x^{h}}{x^{h(h+1)}},
		\end{equation*}
		where last identity is valid because $x^{2h+1}=1$. Hence $-\frac{1+x^h}{x^{h(h+1)}}\in\mathbb{F}_q$, and so $\mathbf{c}\in\mathbb{F}_q^{q+1}$. Next we will demonstrate that $\mathbf{c}\in\mathcal{C}_{(q,q+1,3,h)}$. 
		It is easy to check that 
		\begin{equation}\label{sec3 thm1 equ1}
			1+x^h-\frac{1+x^h}{x^{h(h+1)}}x^{h(h+1)}=0.
		\end{equation}
		Raising to the $q$-th power both sides of this equation yields 
		\begin{equation}\label{sec3 thm1 equ2}
			1+x^{-h}-\frac{1+x^h}{x^{h(h+1)}}x^{-h(h+1)}=1+x^{h+1}-\frac{1+x^h}{x^{h(h+1)}}x^{(h+1)(h+1)}=0.
		\end{equation}
		Recall that $H$ is a parity-check matrix of $\mathcal{C}_{(q,q+1,3,h)}$, then
		\begin{equation*}
			\begin{aligned}
				\mathbf{c}H^T=&\left[\begin{array}{ll}
					1+(\beta^{l_1})h-\frac{1+x^h}{x^{h(h+1)}}(\beta^{l_2})^{h}, & 1+(\beta^{l_1})^{h+1}-\frac{1+x^h}{x^{h(h+1)}}(\beta^{l_2})^{h+1}
				\end{array}\right]\\
				=&\left[\begin{array}{ll}
					1+x^h-\frac{1+x^h}{x^{h(h+1)}}x^{h(h+1)}, & 1+x^{h+1}-\frac{1+x^h}{x^{h(h+1)}}x^{(h+1)(h+1)}
				\end{array}\right]\\
				=&\left[\begin{array}{ll}
					0, & 0
				\end{array}\right]
			\end{aligned}			
		\end{equation*}
		because of \cref{sec3 thm1 equ1,sec3 thm1 equ2}. That is, $\mathbf{c}$ is a codeword of $\mathcal{C}_{(q,q+1,3,h)}$. Consequently, $d=3$.
		
		Conversely, assume that $d={3}$. Then there are three pairwise distinct elements $x, y, z \in U_{q+1}$ such that
		\begin{equation}
			\label{sec3 thm1 equ3}
			\left.i\left[\begin{array}{l}x^{h}\\x^{h+1}\end{array}\right.\right]+j\left[\begin{array}{l}y^{h}\\y^{h+1}\end{array}\right]+k\left[\begin{array}{l}z^{h}\\z^{h+1}\end{array}\right]=\mathbf{0},
		\end{equation}
		where $i,j,k \in \mathbb{F}_{q}^\ast$. Raising to the $q$-th power both sides of \cref{sec3 thm1 equ3} yields 
		\begin{equation}
			\label{sec3 thm1 equ4}
			\left.i\left[\begin{array}{l}x^{-h}\\x^{-(h+1)}\end{array}\right.\right]+j\left[\begin{array}{l}y^{-h}\\y^{-(h+1)}\end{array}\right]+k\left[\begin{array}{l}z^{-h}\\z^{-(h+1)}\end{array}\right]=\mathbf{0},
		\end{equation}
		Combine \cref{sec3 thm1 equ3,sec3 thm1 equ4}, we get
		\begin{equation*}
			\left.i\left[\begin{array}{l}x^{-h}\\x^h\\x^{h+1}\end{array}\right.\right]+j\left[\begin{array}{l}y^{-h}\\y^h\\y^{{h+1}}\end{array}\right]+k\left[\begin{array}{l}z^{-h}\\z^h\\z^{{h+1}}\end{array}\right]=\mathbf{0}
		\end{equation*}
		and
		\begin{equation*}
			\left.i\left[\begin{array}{l}x^{-(h+1)}\\x^h\\x^{h+1}\end{array}\right.\right]+j\left[\begin{array}{l}y^{-(h+1)}\\y^h\\y^{{h+1}}\end{array}\right]+k\left[\begin{array}{c}z^{-(h+1)}\\z^h\\z^{{h+1}}\end{array}\right]=\mathbf{0}.
		\end{equation*}
		Then we have
		\begin{equation*}
			\left|\begin{array}{lll}
				x^{-h} & y^{-h} & z^{-h}\\ x^h & y^h & z^h\\ x^{h+1} & y^{h+1} & z^{h+1}
			\end{array}\right|
			=\left|\begin{array}{lll}x^{-(h+1)} & y^{-(h+1)} & z^{-(h+1)}\\ x^h & y^h & z^h\\ x^{h+1} & y^{h+1} & z^{h+1}\end{array}\right|
			=0.
		\end{equation*}
		According to \cref{sec3 lem2,sec3 lem3}, we have either $x^{2h+1} = y^{2h+1} = z^{2h+1}$, or $x^{2h} = y^{2h} = z^{2h}$ and  $x^{2h+2} = y^{2h+2} = z^{2h+2}$. If the latter holds, then $\left(\frac{x}{y}\right)^{2h}=\left(\frac{z}{y}\right)^{2h}=\left(\frac{x}{y}\right)^{2h+2}=\left(\frac{z}{y}\right)^{2h+2}=1$. Since $\gcd(2h,2h+2)=2$, $\frac{x}{y}\ne1$ and $\frac{z}{y}\ne1$, we have $\frac{x}{y}=\frac{z}{y}=-1$, which would contradicts $x\ne z$. Therefore we have $x^{2h+1} = y^{2h+1} = z^{2h+1}$, which implies $\gcd(2h+1,q+1)>1$ from \cref{sec3 lem4}. The proof of the first part is completed. 
		
		According to the first part, $d \ne 3$ if and only if $\gcd(2h+1,q+1)=1$. Combine the BCH bound and Singleton bound, we have $4 \le d \le 5$ if and only if $d\ne3$. This completes the proof.
	\end{proof}
	
	Combine \cref{sec3 thm0,sec3 thm1}, we have the following corollary about MDS codes.
	\begin{corollary}\label{sec3 cor1}
		(1). For $q$ even, if $h=\frac{q}{2}$, then $\mathcal{C}_{(q,q+1,3,h)}$ is an MDS code with parameters $[q+1,q-1,3]$.
		
		(2). For $q$ even, if $h=0\ \mathrm{or}\ q$, then $\mathcal{C}_{(q,q+1,3,h)}$ is an MDS code with parameters $[q+1,q-2,4]$.
		
		(3). For $q$ odd, if $h=0,\frac{q-1}{2},\frac{q+1}{2}\ \mathrm{or}\ q$, then $\mathcal{C}_{(q,q+1,3,h)}$ is an MDS code with parameters $[q+1,q-2,4]$.
	\end{corollary}
	\begin{proof}
		Suppose that the dimension and minimum distance of $\mathcal{C}_{(q,q+1,3,h)}$ is $k$ and $d$, respectively. 
		
		(1). If $q$ is even and $h=\frac{q}{2}$, it follows from \cref{sec3 thm0} that $k=q-1$. Since $\gcd(2h+1,q+1)=q+1>1$, \cref{sec3 thm1} leads to $d=3$.
		
		(2). If $q$ is even and $h=0\ \mathrm{or}\ q$, by \cref{sec3 thm0}, we have $k=q-2$. Note that $\gcd(1,q+1)=\gcd(2q+1,q+1)=1$, thus $d\ge4$ in accordance with \cref{sec3 thm1}. Then the Singleton bound yields $d\le (q+1)-(q-2)+1=4$, i.e. $d=4$.
		
		(3). If $q$ is odd and $h=0,\frac{q-1}{2},\frac{q+1}{2}\ \mathrm{or}\ q$, we get $k=q-2$ from \cref{sec3 thm0}. One can effortlessly verify that $\gcd(1,q+1)=\gcd(q,q+1)=\gcd(q+2,q+1)=\gcd(2q+1,q+1)=1$. Through \cref{sec3 thm1} and the Singleton bound, we deduce that $d=4$.
		
		It is obvious that the codes in these cases meet the Singleton bound, hence all of these codes are MDS. The proof is completed.
	\end{proof}
	
	If $h\notin\{0,\frac{q-1}{2},\frac{q}{2},\frac{q+1}{2},q\}$, then the dimension of $\mathcal{C}_{(q,q+1,3,h)}$ is $q-3$ by \cref{sec3 thm0}. Thus we have the following corollary which is straightforward from \cref{sec3 thm1}.
	\begin{corollary}\label{sec3 cor2}
		Suppose that $h\notin\{0,\frac{q-1}{2},\frac{q}{2},\frac{q+1}{2},q\}$. Then the BCH code $\mathcal{C}_{(q,q+1,3,h)}$ has parameters $[q+1,q-3,3]$ if and only if $\gcd(2h+1,q+1)>1$. Equivalently, $\mathcal{C}_{(q,q+1,3,h)}$ has parameters $[q+1,q-3,d]$, where $4 \le d \le 5$ if and only if $\gcd(2h+1,q+1)=1$.
	\end{corollary}
	
	\begin{remark}
		We can infer from \cref{sec3 cor2} that the code $\mathcal{C}_{(q,q+1,3,h)}$ must be either an AMDS code or an MDS code when $h\notin\{0,\frac{q-1}{2},\frac{q}{2},\frac{q+1}{2},q\}$ and $\gcd(2h+1,q+1)=1$. It should be noticed that the case of $h=1$ was shown in \cite{Ding2020InfiniteFamiliesMDS}, which indicates that $\mathcal{C}_{(q,q+1,3,1)}$ has parameters $[q+1,q-3,d]$, where $4\le d\le 5$ if and only if 3 does not divide $q+1$. Actually, $3\nmid q+1$ is equivalent to $\gcd(3,q+1)=1$.
	\end{remark}
	
	Before establishing the necessary and sufficient conditions for the minimum distance of $\mathcal{C}_{(q,q+1,3,h)}$ to be 4, we need the following two lemmas. 
	\begin{lemma}\label{sec3 lem5}
		Define
		\begin{equation}\label{sec3 lem5 equ1}
			D(x,y)=\left|\begin{array}{ll}
				x^h & y^h\\
				x^{h+1} & y^{h+1}
			\end{array}\right|=x^hy^h(y-x)
		\end{equation}
		where $x,y\in\mathbb{F}_{q^2}$. If $x,y\in U_{q+1}$, then $D(x,y)^q=-x^{-2h-1}y^{-2h-1}D(x,y)$.
	\end{lemma}
	\begin{proof}
		If $x,y\in U_{q+1}$, then
		\begin{equation*}
			\begin{aligned}
				&D(x,y)^q=x^{qh}y^{qh}(y^q-x^q)=x^{-h}y^{-h}(y^{-1}-x^{-1})\\
				=&-x^{-h-1}y^{-h-1}(y-x)=-x^{-2h-1}y^{-2h-1}x^hy^h(y-x)\\
				=&-x^{-2h-1}y^{-2h-1}D(x,y).
			\end{aligned}
		\end{equation*}
	\end{proof}
	
	\begin{lemma}\label{sec3 lem6}
		Define
		\begin{equation}\label{sec3 lem6 equ1}
			E(x,y)=\frac{x^{2h+1}-y^{2h+1}}{x-y}
		\end{equation}
		where $x,y\in\mathbb{F}_{q^2}$. If $x,y\in U_{q+1}$, then $E(x,y)^q=x^{-2h}y^{-2h}E(x,y)$.
	\end{lemma}
	\begin{proof}
		If $x,y\in U_{q+1}$, then
		\begin{equation*}
			\begin{aligned}
				E(x,y)^q=\frac{x^{-2h-1}-y^{-2h-1}}{x^{-1}-y^{-1}}=\frac{x^{-2h-1}y^{-2h-1}(x^{2h+1}-y^{2h+1})}{x^{-1}y^{-1}(x-y)}=x^{-2h}y^{-2h}E(x,y).
			\end{aligned}
		\end{equation*}
	\end{proof}
		
	Now the following theorem provides the necessary and sufficient conditions for the minimum distance of $\mathcal{C}_{(q,q+1,3,h)}$ to be 4.
	\begin{theorem}\label{sec3.1 thm2}
		The minimum distance $d$ of $\mathcal{C}_{(q,q+1,3,h)}$ is equal to 4 if and only if $\gcd(2h+1,q+1)=1$, and there exist four pairwise distinct elements $x,y,z,w\in U_{q+1}$ such that
		\begin{equation}\label{sec3.1 thm2 equ0}
			\frac{E(x,z)}{E(x,w)}=\frac{E(y,z)}{E(y,w)},
		\end{equation}
		where $E(\ast,\ast)$ is defined as \cref{sec3 lem6 equ1}.
	\end{theorem}
	\begin{proof}
		Recall from the proof of \cref{sec3 thm1} that $\alpha$ is a generator of $\mathbb{F}_{q^2}^{\ast}$, $\beta=\alpha^{q-1}$ is a primitive $(q+1)$-th root of unity in $\mathbb{F}_{q^2}$, and
		\begin{equation*}
			\left.H=\left[\begin{array}{lllll}1&\beta^{h}&(\beta^{h})^2&\cdots&(\beta^{h})^q\\1&\beta^{h+1}&(\beta^{h+1})^2&\cdots&(\beta^{h+1})^q\end{array}\right.\right]
		\end{equation*}
		is a parity-check matrix of $\mathcal{C}_{(q,q+1,3,h)}$. Then $d=4$ if and only if there exist four pairwise distinct elements $x,y,z,w\in U_{q+1}$ such that
		\begin{equation}
			\label{sec3.1 thm2 equ1}
			\left.i\left[\begin{array}{l}x^{h}\\x^{h+1}\end{array}\right.\right]+j\left[\begin{array}{l}y^{h}\\y^{h+1}\end{array}\right]-k\left[\begin{array}{l}z^{h}\\z^{h+1}\end{array}\right]-l\left[\begin{array}{l}w^{h}\\w^{h+1}\end{array}\right]=\mathbf{0},
		\end{equation}
		where $i,j,k,l\in\mathbb{F}_q^\ast$. \cref{sec3.1 thm2 equ1} is equivalent to
		\begin{equation}
			\label{sec3.1 thm2 equ2}
			\left.\left[\begin{array}{ll}x^{h} & y^{h}\\x^{h+1} & y^{h+1}\end{array}\right.\right]\left[\begin{array}{l}i\\j\end{array}\right]=k\left[\begin{array}{l}z^{h}\\z^{h+1}\end{array}\right]+l\left[\begin{array}{l}w^{h}\\w^{h+1}\end{array}\right].
		\end{equation}
		
		 Assume that $d=4$. Then there exist four pairwise distinct elements $x,y,z,w\in U_{q+1}$ and $i,j,k,l\in\mathbb{F}_q^\ast$, such that \cref{sec3.1 thm2 equ2} holds. Since $x\ne y$, it follows from \cref{sec3 lem5 equ1} that $D(x,y)\ne0$. An application of Cramer's Rule yields
		 \begin{equation*}
		 	i=\frac{\left|\begin{array}{ll}kz^{h}+lw^{h} & y^{h}\\kz^{h+1}+lw^{h+1}&y^{h+1}\end{array}\right|}{D(x,y)}=\frac{kD(z,y)+lD(w,y)}{D(x,y)}
		 \end{equation*}
		 and
		 \begin{equation*}
		 	j=\frac{\left|\begin{array}{ll}x^{h} &  kz^{h}+lw^{h}\\x^{h+1} & kz^{h+1}+lw^{n+1}\end{array}\right|}{D(x,y)}=\frac{kD(x,z)+lD(x,w)}{D(x,y)}.
		 \end{equation*}
		 Note that $k,l\in \mathbb{F}_q^\ast$, then $k^q=k$ and $l^q=l$. According to \cref{sec3 lem5}, we have
		 \begin{equation*}
		 	i^{q}=k\frac{D(z,y)}{D(x,y)}\left(\frac{x}{z}\right)^{2h+1}+l\frac{D(w,y)}{D(x,y)}\left(\frac{x}{w}\right)^{2h+1}
		 \end{equation*}
		 and 
		 \begin{equation*}
		 	j^{q}=k\frac{D(x,z)}{D(x,y)}\left(\frac{y}{z}\right)^{2h+1}+l\frac{D(x,w)}{D(x,y)}\left(\frac{y}{w}\right)^{2h+1}.
		 \end{equation*}
		 Since $i,j\in\mathbb{F}_q^\ast$, $i^q=i$, which implies
		 \begin{equation*}
		 	k\frac{D(z,y)}{D(x,y)}\left(\frac{x}{z}\right)^{2h+1}+l\frac{D(w,y)}{D(x,y)}\left(\frac{x}{w}\right)^{2h+1}=\frac{kD(z,y)+lD(w,y)}{D(x,y)}.
		 \end{equation*}
		 Rearranging terms, we get
		 \begin{equation}\label{sec3.1 thm2 equ3}
		 	k\frac{D(z,y)}{D(x,y)}\left(\left(\frac{x}{z}\right)^{2h+1}-1\right)=-l\frac{D(w,y)}{D(x,y)}\left(\left(\frac{x}{w}\right)^{2h+1}-1\right).
		 \end{equation}
		 It follows from \cref{sec3 thm1} that $\gcd(2h+1,q+1)=1$. Since $x,y,z,w$ are pairwise distinct, $\left(\frac{x}{z}\right)^{2h+1}-1\ne0$ and $\left(\frac{x}{w}\right)^{2h+1}-1\ne0$. Note that $D(z,y)\ne0$ and $D(w,y)\ne0$. Hence neither side of \cref{sec3.1 thm2 equ3} is equal to 0. Then
		 \begin{equation}\label{sec3.1 thm2 equ4}
		 	\begin{aligned}
		 		&-\frac{k}{l}=\frac{\frac{D(w,y)}{D(x,y)}\left(\left(\frac{x}{w}\right)^{2h+1}-1\right)}{\frac{D(z,y)}{D(x,y)}\left(\left(\frac{x}{z}\right)^{2h+1}-1\right)}=\frac{D(w,y)\left(\left(\frac{x}{w}\right)^{2h+1}-1\right)}{D(z,y)\left(\left(\frac{x}{z}\right)^{2h+1}-1\right)}\\
		 		=&\frac{w^hy^h(y-w)\left(\left(\frac{x}{w}\right)^{2h+1}-1\right)}{z^hy^h(y-z)\left(\left(\frac{x}{z}\right)^{2h+1}-1\right)}=\frac{w^{h+1}(\frac{y}{w}-1)\left(\left(\frac{x}{w}\right)^{2h+1}-1\right)}{z^{h+1}(\frac{y}{z}-1)\left(\left(\frac{x}{z}\right)^{2h+1}-1\right)}.
		 	\end{aligned}
		 \end{equation}
		 Similarly, by $j^q=j$, we obtain
		 \begin{equation}\label{sec3.1 thm2 equ5}
		 	\begin{aligned}
		 		&-\frac{k}{l}=\frac{\frac{D(x,w)}{D(x,y)}\left(\left(\frac{y}{w}\right)^{2h+1}-1\right)}{\frac{D(x,z)}{D(x,y)}\left(\left(\frac{y}{z}\right)^{2h+1}-1\right)}=\frac{D(x,w)\left(\left(\frac{y}{w}\right)^{2h+1}-1\right)}{D(x,z)\left(\left(\frac{y}{z}\right)^{2h+1}-1\right)}\\
		 		=&\frac{x^hw^h(w-x)\left(\left(\frac{y}{w}\right)^{2h+1}-1\right)}{x^hz^h(z-x)\left(\left(\frac{y}{z}\right)^{2h+1}-1\right)}=\frac{w^{h+1}(\frac{x}{w}-1)\left(\left(\frac{y}{w}\right)^{2h+1}-1\right)}{z^{h+1}(\frac{x}{z}-1)\left(\left(\frac{y}{z}\right)^{2h+1}-1\right)}.
		 	\end{aligned}
		 \end{equation}
		 By combining \cref{sec3.1 thm2 equ4,sec3.1 thm2 equ5}, we have
		 \begin{equation}\label{sec3.1 thm2 equ6}
		 	\frac{(\frac{y}{w}-1)\left(\left(\frac{x}{w}\right)^{2h+1}-1\right)}{(\frac{y}{z}-1)\left(\left(\frac{x}{z}\right)^{2h+1}-1\right)}=\frac{(\frac{x}{w}-1)\left(\left(\frac{y}{w}\right)^{2h+1}-1\right)}{(\frac{x}{z}-1)\left(\left(\frac{y}{z}\right)^{2h+1}-1\right)}.
		 \end{equation}
		 Multiplying both sides of \cref{sec3.1 thm2 equ6} by $\frac{w^{2h+2}}{z^{2h+2}}$ yields
		 \begin{equation*}
		 	\frac{(y-w)\left({x}^{2h+1}-{w}^{2h+1}\right)}{(y-z)\left({x}^{2h+1}-{z}^{2h+1}\right)}=\frac{(x-w)\left({y}^{2h+1}-{w}^{2h+1}\right)}{(x-z)\left({y}^{2h+1}-{z}^{2h+1}\right)}.
		 \end{equation*}
		 This equation is the same as
		 \begin{equation}\label{sec3.1 thm2 equ7}
		 	\frac{\frac{{x}^{2h+1}-{z}^{2h+1}}{x-z}}{\frac{{x}^{2h+1}-{w}^{2h+1}}{x-w}}=\frac{\frac{{y}^{2h+1}-{z}^{2h+1}}{y-z}}{\frac{{y}^{2h+1}-{w}^{2h+1}}{y-w}},
		 \end{equation}
		 i.e. $\frac{E(x,z)}{E(x,w)}=\frac{E(y,z)}{E(y,w)}$.
		 
		 Conversely, assume that $\gcd(2h+1,q+1)=1$, and there exist four pairwise distinct elements $x,y,z,w\in U_{q+1}$ such that $\frac{E(x,z)}{E(x,w)}=\frac{E(y,z)}{E(y,w)}$. Note that \cref{sec3.1 thm2 equ6,sec3.1 thm2 equ7} are equivalent. According to \cref{sec3.1 thm2 equ4,sec3.1 thm2 equ5}, multiplying both sides by $\frac{w^{h+1}}{z^{h+1}}$ yields
		 \begin{equation*}
			\frac{D(w,y)\left(\left(\frac{x}{w}\right)^{2h+1}-1\right)}{D(z,y)\left(\left(\frac{x}{z}\right)^{2h+1}-1\right)}=\frac{D(x,w)\left(\left(\frac{y}{w}\right)^{2h+1}-1\right)}{D(x,z)\left(\left(\frac{y}{z}\right)^{2h+1}-1\right)}.
		 \end{equation*}
		 Since $D(x,y)\ne0$, 
		 \begin{equation}\label{sec3.1 thm2 equ8}
		 	\frac{\frac{D(w,y)}{D(x,y)}\left(\left(\frac{x}{w}\right)^{2h+1}-1\right)}{\frac{D(z,y)}{D(x,y)}\left(\left(\frac{x}{z}\right)^{2h+1}-1\right)}=\frac{\frac{D(x,w)}{D(x,y)}\left(\left(\frac{y}{w}\right)^{2h+1}-1\right)}{\frac{D(x,z)}{D(x,y)}\left(\left(\frac{y}{z}\right)^{2h+1}-1\right)}.
		 \end{equation}
		 Suppose that $l=-1$, and 
		 \begin{equation*}
		 	k=\frac{\frac{D(w,y)}{D(x,y)}\left(\left(\frac{x}{w}\right)^{2h+1}-1\right)}{\frac{D(z,y)}{D(x,y)}\left(\left(\frac{x}{z}\right)^{2h+1}-1\right)}=\frac{D(w,y)\left(\left(\frac{x}{w}\right)^{2h+1}-1\right)}{D(z,y)\left(\left(\frac{x}{z}\right)^{2h+1}-1\right)},
		 \end{equation*}
		 Since $\gcd(2h+1,q+1)=1$ and $x,y,z,w\in U_{q+1}$ are different from each other, we have $k\ne0$. In addition, we have
		 \begin{equation*}
		 	\begin{aligned}
		 		k^q=&\left(\frac{D(w,y)\left(\left(\frac{x}{w}\right)^{2h+1}-1\right)}{D(z,y)\left(\left(\frac{x}{z}\right)^{2h+1}-1\right)}\right)^q\\
		 		=&\left(\frac{-w^{-2h-1}y^{-2h-1}D(w,y)\left(\left(\frac{x}{w}\right)^{-2h-1}-1\right)}{-z^{-2h-1}y^{-2h-1}D(z,y)\left(\left(\frac{x}{z}\right)^{-2h-1}-1\right)}\right)^q\\
		 		=&\left(\frac{-w^{-2h-1}y^{-2h-1}D(w,y)\cdot\left(-\left(\frac{w}{x}\right)^{2h+1}\right)\left(\left(\frac{x}{w}\right)^{2h+1}-1\right)}{-z^{-2h-1}y^{-2h-1}D(z,y)\cdot\left(-\left(\frac{z}{x}\right)^{2h+1}\right)\left(\left(\frac{x}{z}\right)^{2h+1}-1\right)}\right)^q\\
		 		=&k.
		 	\end{aligned}
		 \end{equation*}
		 Thus $k\in\mathbb{F}_q^\ast$. Let
		 \begin{equation*}
		 	i=\frac{kD(z,y)+lD(w,y)}{D(x,y)}\ \mathrm{and}\ j=\frac{kD(x,z)+lD(x,w)}{D(x,y)}.
		 \end{equation*}
		 Next we will show that $i,j\in\mathbb{F}_q^\ast$. Substituting $k$ and $l$ into $i^q-i$ leads to 
		 \begin{equation*}
		 	i^q-i=k\frac{D(z,y)}{D(x,y)}\left(\left(\frac{x}{z}\right)^{2h+1}-1\right)+l\frac{D(w,y)}{D(x,y)}\left(\left(\frac{x}{w}\right)^{2h+1}-1\right)=0,
		 \end{equation*}
		 i.e. $i\in\mathbb{F}_q$. Note that
		 \begin{equation*}
		 	\begin{aligned}
		 		iD(x,y)=&kD(z,y)+lD(w,y)\\
		 		=&\frac{D(w,y)\left(\left(\frac{x}{w}\right)^{2h+1}-1\right)}{D(z,y)\left(\left(\frac{x}{z}\right)^{2h+1}-1\right)}D(z,y)-D(w,y)\\
		 		=&\frac{\left(\left(\frac{x}{w}\right)^{2h+1}-1\right)-\left(\left(\frac{x}{z}\right)^{2h+1}-1\right)}{\left(\frac{x}{z}\right)^{2h+1}-1}D(w,y)\\
		 		=&\frac{\left(\frac{x}{w}\right)^{2h+1}-\left(\frac{x}{z}\right)^{2h+1}}{\left(\frac{x}{z}\right)^{2h+1}-1}D(w,y)\\
		 		\ne&0,
		 	\end{aligned}
		 \end{equation*}
		 the last inequality follows from $\gcd(2h+1,q+1)=1$ and $w\ne z$. Then $i\ne 0$ as $D(x,y)\ne 0$, and so $i\in\mathbb{F}_q^\ast$. By \cref{sec3.1 thm2 equ8}, we also have 
		 \begin{equation*}
		 	k=\frac{\frac{D(x,w)}{D(x,y)}\left(\left(\frac{y}{w}\right)^{2h+1}-1\right)}{\frac{D(x,z)}{D(x,y)}\left(\left(\frac{y}{z}\right)^{2h+1}-1\right)}=\frac{D(x,w)\left(\left(\frac{y}{w}\right)^{2h+1}-1\right)}{D(x,z)\left(\left(\frac{y}{z}\right)^{2h+1}-1\right)}.
		 \end{equation*}
		 Then one can deduce that $j\in\mathbb{F}_q^\ast$ similarly. Applying Cramer's Rule again, there exist $i,j,k,l\in\mathbb{F}_q^\ast$ and four pairwise distinct elements $x,y,z,w\in U_{q+1}$ such that \cref{sec3.1 thm2 equ2} holds. Equivalently, \cref{sec3.1 thm2 equ1} is satisfied. Therefore, $d=4$
		 
		 This completes the proof.
	\end{proof}
	
	Combine \cref{sec3 thm0,sec3.1 thm2}, we present the following corollary that describes the necessary and sufficient conditions for $\mathcal{C}_{(q,q+1,3,h)}$ to be an AMDS code.
	\begin{corollary}\label{sec3 cor4}
		Suppose that $h\notin\{0,\frac{q-1}{2},\frac{q}{2},\frac{q+1}{2},q\}$. Then the BCH code $\mathcal{C}_{(q,q+1,3,h)}$ is an AMDS code with parameters $[q+1,q-3,4]$ if and only if $\gcd(2h+1,q+1)=1$ and there exist four pairwise distinct elements $x,y,z,w\in U_{q+1}$ such that $\frac{E(x,z)}{E(x,w)}=\frac{E(y,z)}{E(y,w)}$.
	\end{corollary}
	
	Actually, if $q$ is odd and $\gcd(2h+1,q+1)=1$, the appropriate $x,y,z,w\in U_{q+1}$ in \cref{sec3.1 thm2} always exist.
	\begin{theorem}\label{sec3.1 thm3}
		Suppose that $q$ is odd. Then minimum distance $d$ of the BCH code $\mathcal{C}_{(q,q+1,3,h)}$ is equal to 4 if and only if  $\gcd(2h+1,q+1)=1$.
	\end{theorem}
	\begin{proof}
		According to \cref{sec3.1 thm2}, we only need to prove that there exist four pairwise distinct elements $x,y,z,w\in U_{q+1}$ such that $\frac{E(x,z)}{E(x,w)}=\frac{E(y,z)}{E(y,w)}$. Since $q$ is odd, then $1,-1\in U_{q+1}$ are different elements and the cardinality of $U_{q+1}$ is at least 4. Suppose that $x\in U_{q+1}\setminus\{1,-1\}$. Then we have $x^{-1}\in U_{q+1}$ and $x,x^{-1},1,-1$ are pairwise distinct. Substituting $(x,x^{-1},1,-1)$ for $(x,y,z,w)$ into $\frac{E(y,z)}{E(y,w)}$ yields
		\begin{equation*}
			\frac{E(y,z)}{E(y,w)}=\frac{E(x^{-1},1)}{E(x^{-1},-1)}=\frac{\frac{x^{-2h-1}-1}{x^{-1}-1}}{\frac{x^{-2h-1}+1}{x^{-1}+1}}=\frac{\frac{x^{2h+1}-1}{x-1}}{\frac{x^{2h+1}+1}{x+1}}=\frac{E(x,1)}{E(x,-1)}=\frac{E(x,z)}{E(x,w)}.
		\end{equation*}
		This proof is completed.
	\end{proof}
	
	When $q$ is odd, the following corollary shows the necessary and sufficient conditions for $\mathcal{C}_{(q,q+1,3,h)}$ to be an AMDS code. Its proof is omitted as it can be derived directly from \cref{sec3 thm0,sec3.1 thm3}.
	\begin{corollary}\label{sec3 cor5}
		Suppose that $q$ is odd and $h\notin\{0,\frac{q-1}{2},\frac{q+1}{2},q\}$. Then the BCH code $\mathcal{C}_{(q,q+1,3,h)}$ is an AMDS code with parameters $[q+1,q-3,4]$ if and only if $\gcd(2h+1,q+1)=1$.
	\end{corollary}
	
	\begin{remark}\label{sec3 rmk2}
		(1) For $q$ odd, it is worth noting that $\mathcal{C}_{(q,q+1,3,h)}$ cannot be an MDS code with dimension $q-3$ in accordance with \cref{sec3 thm1,sec3 cor2,sec3 cor5}. Therefore, the condition that $q$ is even is necessary for $\mathcal{C}_{(q,q+1,3,h)}$ to be an MDS code with dimension $q-3$.
		
		(2) For fixed $h$ and $q$, it is easily checked whether $gcd(2h+1,q+1)=1$ holds. Hence one can determine the parameters of $\mathcal{C}_{(q,q+1,3,h)}$ when $q$ is odd. However, for $q$ even, it is not easy to obtain the parameters of $\mathcal{C}_{(q,q+1,3,h)}$.
	\end{remark}
	
	We will now examine the minimum distance $d^\perp$ of the code $\mathcal{C}_{(q,q+1,3,h)}^\perp$ and present a sufficient condition for $\mathcal{C}_{(q,q+1,3,h)}$ to be AMDS.
	\begin{theorem}\label{sec3 thm4}
		Let $m$ be the larger of $\gcd(2h,q+1)$ and $\gcd(2h+2,q+1)$. Suppose that $h\notin\{0,\frac{q-1}{2},\frac{q}{2},\frac{q+1}{2},q\}$. Then $\mathcal{C}_{(q,q+1,3,h)}^\perp$ has parameters $[q+1,4,d^\perp]$ with $q-2h-1 \le d^\perp \le q+1-m$. Furthermore, if $\gcd(2h+1,q+1)=1$ and $m\ge4$, then the BCH code $\mathcal{C}_{(q,q+1,3,h)}$ is an AMDS code with parameters $[q+1,q-3,4]$.
	\end{theorem}
	\begin{proof}
		The dimension of $\mathcal{C}_{(q,q+1,3,h)}^\perp$ is $4$ by \cref{sec3 cor0}. We recall from the proof of \cref{sec3 thm1} that $\alpha$ is a generator of $\mathbb{F}_{q^2}^{\ast}$, $\beta=\alpha^{q-1}$ is a primitive $(q+1)$-th root of unity in $\mathbb{F}_{q^2}$ and
		\begin{equation*}
			\left.H=\left[\begin{array}{lllll}1&\beta^{h}&(\beta^{h})^2&\cdots&(\beta^{h})^q\\1&\beta^{h+1}&(\beta^{h+1})^2&\cdots&(\beta^{h+1})^q\end{array}\right.\right]
		\end{equation*}
		is a parity-check matrix of $\mathcal{C}_{(q,q+1,3,h)}$. Then it follows from Delsarte’s theorem that the trace expression of $\mathcal{C}_{(q,q+1,3,h)}^{\perp}$ is given by
		\begin{equation*}
			\mathcal{C}_{(q,q+1,3,h)}^\perp=\{\mathbf{c}_{(a,b)}:a,b\in\mathbb{F}_{q^2}\},
		\end{equation*}
		where $\mathbf{c}_{(a,b)}=\left(\mathrm{Tr}_{q^{2}/q}\left(a\beta^{hi}+b\beta^{(h+1)i}\right)\right)_{i=0}^{q}.$ Let $u \in U_{q+1}$. Then 
		\begin{equation}
			\label{sec3 thm4 equ1}
			\begin{aligned}
				\operatorname{Tr}_{q^2/q}(au^h+bu^{h+1})&=au^h+bu^{h+1}+a^qu^{-h}+b^qu^{-(h+1)}\\&=u^{-(h+1)}(bu^{2h+2}+au^{2h+1}+a^qu+b^q).
			\end{aligned}
		\end{equation}
		Now consider the number of solutions of the equation
		\begin{equation}
			\label{sec3 thm4 equ2}
			bu^{2h+2}+au^{2h+1}+a^qu+b^q=0.
		\end{equation}
		When $(a,b) \ne (0,0)$, this equation has at most $2h+2$ solutions in $U_{q+1}$, that is, we have
		\begin{equation*}
			wt(\mathbf{c}_{(a,b)})\geq q+1-2h-2=q-2h-1, 
		\end{equation*}
		i.e., $d^\perp \ge q-2h-1$.
		
		If $m=\gcd(2h+2,q+1)$, substituting $(a,b)$ by $(0,\alpha^{q+1})$ if $p$ is even and by $(0,\alpha^{\frac{q+1}{2}})$ if $p$ is odd in \cref{sec3 thm4 equ2} leads to $b(u^{2h+2}-1)=0$. Since $\gcd(2h+2,q+1)=m$, there are $m$ $u \in U_{q+1}$ such that $u^{2h+2}-1=0$. Consequently, the Hamming weight of $\mathbf{c}_{(a,b)}$ is $q+1-m$ when $(a,b) = (0,\alpha^{q+1})$ if $p$ is even and $(a,b) = (0,\alpha^{\frac{q+1}{2}})$ if $p$ is odd. This means that $d^\perp \le q+1-m$.
		
		If $m=\gcd(2h,q+1)$, substituting $(\alpha^{q+1},0)$ if $p$ is even and $(\alpha^{\frac{q+1}{2}},0)$ if $p$ is odd for $(a,b)$ in \cref{sec3 thm4 equ2} yields $au(u^{2h}-1)=0$. Since $\gcd(2h,q+1)=m$, there exist $m$ $u \in U_{q+1}$ such that $u^{2h}-1=0$. Then the Hamming weight of $\mathbf{c}_{(a,b)}$ is $q+1-m$ when $(a,b) = (\alpha^{q+1},0)$ if $p$ is even and $(a,b) = (\alpha^{\frac{q+1}{2}},0)$ if $p$ is odd. Hence, we have $d^\perp \le q+1-m$. As a result, $q-2h-1 \le d^\perp \le q+1-m$.
		
		Next we demonstrate that the minimum distance $d$ of $\mathcal{C}_{(q,q+1,3,h)}$ is 4. It follows from \cref{sec3 thm1} that $4\le d\le5$. We only need to prove that $d\ne5$. If $d = 5$, then both $\mathcal{C}_{(q,q+1,3,h)}$ and $\mathcal{C}_{(q,q+1,3,h)}^\perp$ are MDS. However, we have deduced that $d^\perp \le q+1-m$ in the preceding discussion. Since $m\ge4$, then $d^\perp \le q-3$, and so $\mathcal{C}_{(q,q+1,3,h)}^\perp$ is not MDS. Hence, $d = 4$. This completes the proof.
	\end{proof}
	
	\begin{remark}\label{sec3 rmk3}
		It is easily seen that the condition of $m \ge 4$ in \cref{sec3 thm4} is not necessary for the code $\mathcal{C}_{(q,q+1,3,h)}$ to be an AMDS code from the examples verified by Magma in \cref{sec4 rmk1,sec4 rmk2} in \cref{sec4}.
	\end{remark}
	
	When $q$ is odd and $h=1$, \cref{sec3 cor5,sec3 thm5} immediately lead to the following interesting result.
	\begin{corollary}\label{sec3 cor6}
		Suppose that $q=p^s>4$, where $p$ is odd. If $\gcd(2h+1,q+1)=1$, then $\mathcal{C}_{(q,q+1,3,1)}$ is an NMDS code with parameters $[q+1,q-3,4]$, and its dual code $\mathcal{C}_{(q,q+1,3,1)}^\perp$ has parameters $[q+1,q-3,4]$.
	\end{corollary}
	\begin{remark}\label{sec3 rmk4}
		The case of $p=3$ in \cref{sec3 cor6} has been presented in \cite{Ding2020InfiniteFamiliesMDS}. In fact, \cref{sec3 cor6} extends the relevant results presented in in \cite{Ding2020InfiniteFamiliesMDS}.
	\end{remark}

	In addition, \cref{sec3 thm5} indicates that the AMDS code $\mathcal{C}_{(q,q+1,3,h)}$ in \cref{sec3 thm4} is a $d$-optimal and $k$-optimal LRC when $q>4h$. 
	\begin{theorem}
		\label{sec3 thm5}
		If $q>2h+2$, the BCH code $\mathcal{C}_{(q,q+1,3,h)}$ in \cref{sec3 thm4} is a $(q+1,q-3,4,q;d^\perp-1)$-LRC and its dual code $\mathcal{C}^\perp_{(q,q+1,3,h)}$ is a $(q+1,4,d^\perp,q;3)$-LRC. Furthermore, if $q>4h$, $\mathcal{C}_{(q,q+1,3,h)}$ is both $d$-optimal and $k$-optimal.
	\end{theorem}
	\begin{proof}
		With the symbolism of \cref{sec3 thm4}, the minimum distance $d$ of $\mathcal{C}_{(q,q+1,3,h)}$ is 4 and the minimum distance $d^\perp$ of $\mathcal{C}_{(q,q+1,3,h)}^\perp$ satisfies $q-2h-1 \le d^\perp \le q+1-m$. Since $d=4\ge2$ and $d^\perp\ge q-2h-1\ge2h+3-2h-1=2$, both $\mathcal{C}_{(q,q+1,3,h)}$ and its dual are nontrivial. According to \cref{sec2 lem3}, $\mathcal{C}_{(q,q+1,3,h)}$ and $\mathcal{C}^\perp_{(q,q+1,3,h)}$ are LCRs with locality $d^\perp-1$ and 3, respectively. Since $m\ge4$, we have $d^\perp \le q+1-m \le q-3$, and so $1<\frac{q-3}{d^{\perp}-1}$. Note that $\frac{q-3}{d^{\perp}-1}\le2$ is equivalent to $d^\perp\ge\frac{q-1}{2}$. When $q>4h$, we have $d^\perp\ge q-2h-1 \ge\frac{q-1}{2}$. Therefore, $\lceil\frac{q-3}{d^\perp-1}\rceil=2.$ We get from \cref{sec2 lem1} that
		\begin{equation*}
			\begin{aligned}
				&n-k-\left\lceil{\frac{k}{r}}\right\rceil+2 \\
				=&(q+1)-(q-3)-\left\lceil\frac{q-3}{d^\perp-1}\right\rceil+2 \\
				=&4=d,
			\end{aligned}
		\end{equation*}
		i.e. $\mathcal{C}_{(q,q+1,3,4)}$ is $d$-optimal. Moreover, choose $t=1$ in \cref{sec2 lem2}, we have 
		\begin{equation*}
			\begin{aligned}
				&\operatorname*{min}_{t\in\mathbb{Z}_{+}}\left\{tr+k_{opt}^{(q)}(n-t(r+1),d)\right\}\\
				\leq&(d^{\perp}-1)+k_{opt}^{(q)}(q+1-d^{\perp},4) \\
				\leq&d^{\perp}-1+q-2-d^{\perp} \\
				=&q-3\leq k.
			\end{aligned}
		\end{equation*}
		Hence, $\mathcal{C}_{(q,q+1,3,h)}$ is $k$-optimal.
	\end{proof}
	
	\section{Several examples of AMDS codes}\label{sec4}
	
	In this section, we present several specific classes of AMDS codes that are $d$-optimal and $k$-optimal LRCs based on the theorems in \cref{sec3}.
	
	For $p=3$ and $h=4$, it has been presented in \cite{Geng2022ClassAlmostMDS} that $\mathcal{C}_{(3^s,3^s+1,3,4)}$ is AMDS when $s$ is odd. In fact, when $s$ is not only odd but also even, $\mathcal{C}_{(3^s,3^s+1,3,4)}$ with $s\ge3$ is an AMDS code and a $d$-optimal and $k$-optimal LRC according to \cref{sec3 cor4,sec3 thm5}.
	\begin{theorem}\label{sec4 thm1}
		Suppose that $s\ge3$. The BCH code $\mathcal{C}_{(3^s,3^s+1,3,4)}$ is an AMDS code with parameters $[q+1,q-3,4]$ and its dual has parameters $[q+1,4,d^\perp]$ with $q-9\le d^\perp\le q-3$ if $s$ is odd and $d^\perp=q-9$ if $s$ is even. In addition, $\mathcal{C}_{(3^s,3^s+1,3,4)}$ is a $d$-optimal and $k$-optimal LRC and its dual code is an LRC with locality 3.
	\end{theorem}
	\begin{proof}
		In this case $h=4$, we have $\gcd(9,3^s+1)=1$ and $h=\frac{3^2-1}{2}$. Since $s\ge3$, it is follows from \cref{sec3 cor4} that $\mathcal{C}_{(3^s,3^s+1,3,4)}$ has parameters $[q+1,q-3,4]$. If $s$ is odd, $\gcd(10,3^s+1)=2$ and $\gcd(8,3^s+1)=4$. If $s$ is even, we have $\gcd(8,3^s+1)=2$ and $\gcd(10,3^s+1)=10$. Then \cref{sec3 thm5} immediately yields the remaining part due to $3^s>16$.
	\end{proof}
	
	For the case of odd $p$ and $h=\frac{p-1}{2}$, we display the following theorem.
	
	\begin{theorem}
		\label{sec4 thm2}
		Let $q=p^s$, where $p$ is odd prime and $s \ge 2$. Suppose that $h=\frac{p-1}{2}$. Then $\mathcal{C}_{(q,q+1,3,h)}$ is an AMDS code with parameters $[q+1,q-3,4]$, and its dual code $\mathcal{C}_{(q,q+1,3,h)}^{\perp}$ has parameters $[q+1,4,d^\perp]$, where $d^\perp=q-p$ if $s$ is odd and $q-p\le d^\perp\le q-3$ if $s$ is odd. Simultaneously, $\mathcal{C}_{(q,q+1,3,h)}$ is a $d$-optimal and $k$-optimal LRC and its dual code is an LRC with locality 3.
	\end{theorem}
	\begin{proof}
		It is obvious that $\gcd(2h+1,q+1)=\gcd(p,p^s+1)=1$. Since $q$ is odd, it follows from \cref{sec3 cor5} that $\mathcal{C}_{(q,q+1,3,h)}$ is an AMDS code with parameters $[q+1,q-3,4]$. Thus $d^\perp \le q-3$, otherwise $\mathcal{C}_{(q,q+1,3,h)}^{\perp}$ is MDS, which is impossible. Note that
		\begin{equation*}
			p^s+1=\left(p^{s-1}+p^{s-2}+\cdots+p+1\right)(p-1)+2.
		\end{equation*}
		Then we have $\gcd(2h,q+1)=\gcd(p-1,p^s+1)=\gcd(p-1,2)=2$ as $p$ is odd. When $s$ is odd, we have 
		\begin{equation*}
			p^s+1\equiv(-1)^s+1\equiv0\pmod{p+1},
		\end{equation*}
		which implies $\gcd(2h+2,q+1)=\gcd(p+1,p^s+1)=p+1\ge4$. Since $q=p^s>2p-2=4h>p+1=2h+2$ when $s\ge 3$, then one can deduce the theorem from \cref{sec3 thm4,sec3 thm5}.
	\end{proof}
	
	\begin{remark}\label{sec4 rmk1}
		With the same notation as \cref{sec4 thm2}, when $s$ is even, we get $\gcd(2h,q+1)=\gcd(2h+2,q+1)=2$. Hence we cannot accurately determine the minimum distance of $\mathcal{C}_{(q,q+1,3,h)}^\perp$ from \cref{sec3 thm4}. With the help of Magma program, we verify that $d^\perp=q-p$ when $s=2,5\le p \le 61$ and $s=4,p=5,7$. And this prompts us to propose the following conjecture.
	\end{remark}
	
	\begin{conjecture}
		Let $q=p^s$, where $p$ is odd prime and $s \ge 2$. Suppose that $h=\frac{p-1}{2}$. Then BCH code $\mathcal{C}_{(q,q+1,3,h)}$ is an AMDS code with parameters $[q+1,q-3,4]$ and its dual code has parameters $[q+1,4,q-p]$. Simultaneously, $\mathcal{C}_{(q,q+1,3,h)}$ is a $d$-optimal and $k$-optimal LRC with locality $q-p-1$ and its dual code is an LRC with locality 3.
	\end{conjecture}
	
	For $p=2$, the case of $h=1$ has been presented in \cite{Ding2020InfiniteFamiliesMDS}. For the cases where $h=2,3$, there is no $s$ that satisfies the condition $m\ge4$ in \cref{sec3 thm4}. Hence we will now proceed to examine the case of $p=2$ and $h=4$.
	
	\begin{theorem}
		\label{sec3 cor3}
		Let $q=2^s$ with $s \equiv 2 \pmod 4$ and $s \ge 6$. Then the BCH code $\mathcal{C}_{(q,q+1,3,4)}$ is an AMDS code with parameters $[q+1,q-3,4]$ and a $d$-optimal and $k$-optimal LRC. Its dual code $\mathcal{C}_{(q,q+1,3,4)}^{\perp}$ has parameters $[q+1,4,d^{\perp}]$ with $q-9 \le d^\perp \le q-4$.
	\end{theorem}
	\begin{proof}
		Since $\gcd(9,2^s+1)=1,\gcd(8,2^s+1)=1$ and $\gcd(10,2^s+1)=5$ when $s \equiv 2 \pmod 4$, in accordance with \cref{sec3 thm4} we obtain that $\mathcal{C}_{(q,q+1,3,4)}$ and its dual have parameters $[q+1,q-3,4]$ and $[q+1,4,d^{\perp}]$ with $q-9 \le d^\perp \le q-4$, respectively. It is obvious that $\mathcal{C}_{(q,q+1,3,4)}$ is a $d$-optimal and $k$-optimal LRC because $2^s>16$ for $s\ge6$ from \cref{sec3 thm5}. This completes the proof.
	\end{proof}
	
	\begin{remark}\label{sec4 rmk2}
		When $s\ge6$ and $s\equiv0\pmod4$, we have $\gcd(9,2^s+1)=\gcd(8,2^s+1)=\gcd(10,2^s+1)=1$. According to \cref{sec3 thm1}, $\mathcal{C}_{(2^s,2^s+1,3,4)}$ has parameters $[q+1,q-3,d]$ with $4\le d\le 5$. The code $\mathcal{C}_{(2^s,2^s+1,3,4)}$ is AMDS for $s=8,12$ according to Magma experiments. Hence we conjecture that the code $\mathcal{C}_{(2^s,2^s+1,3,4)}$ is AMDS providing that $s\ge 6$ and $s$ is even.
	\end{remark}
	
	We next consider the case where $p=3$ and $h=2$.
	\begin{theorem}
		\label{sec4 thm4}
		Let $q=3^s$, where $s\not\equiv2\pmod4$ and $s \ge 3$. Then the BCH code $\mathcal{C}_{(q,q+1,3,2)}$ is an AMDS code with parameters $[q+1,q-3,4]$, and its dual code $\mathcal{C}_{(q,q+1,3,2)}^{\perp}$ has parameters $[q+1,4,d^{\perp}]$ with $q-5 \le d^\perp \le q-3$. Moreover, $\mathcal{C}_{(q,q+1,3,2)}$ is a $d$-optimal and $k$-optimal LRC.
	\end{theorem}
	\begin{proof}
		When $s\not\equiv2\pmod4$, it is easily calculated that $\gcd(5,3^s+1)=1$, which implies $\mathcal{C}_{(q,q+1,3,2)}$ is an AMDS code with parameters $[q+1,q-3,4]$. Since $3^s>8>6$ for $s\ge3$, the theorem is shown by \cref{sec3 thm4,sec3 thm5}.
	\end{proof}
	
	The cases of $p=5,h=3$ and $p=7,h=4$ are exhibited in the following two theorems and the proofs are omitted as it is similar to the above.
	
	\begin{theorem}\label{sec4 thm5}
		Suppose that $q=5^s$, $s \not\equiv 3 \pmod 6$ and $s \ge 2$. Then the BCH code $\mathcal{C}_{(q,q+1,3,3)}$ has parameters $[q+1,q-3,4]$, and its dual has parameters $[q+1,4,d^{\perp}]$, where $q-7 \le d^\perp \le q-5$ if $s \equiv 1,5 \pmod 6$ and $q-7 \le d^\perp \le q-3$ if $s$ is even. In addition, $\mathcal{C}_{(q,q+1,3,3)}$ is a $d$-optimal and $k$-optimal LRC.
	\end{theorem}
	
	\begin{theorem}
		\label{sec4 thm6}
		Assume that $q=7^s$ with $s \ge 2$. Then the BCH code $\mathcal{C}_{(q,q+1,3,4)}$ is an AMDS code with parameters $[q+1,q-3,4]$ and a $d$-optimal and $k$-optimal LRC. Its dual code has parameters $[q+1,4,d^{\perp}]$, where $d^{\perp}=q-9$ if $s\equiv2\pmod4$, $q-9 \le d^\perp \le q-7$ if $s$ is odd, and $q-9 \le d^\perp \le q-3$ is $s\equiv0\pmod4$.
	\end{theorem}
	
	\section{Concluding remarks}\label{sec5}
	
	In this paper, we demonstrate the sufficient and necessary conditions for the minimum distance $d$ of the BCH code $\mathcal{C}_{(q,q+1,3,h)}$ to be 3 and 4, respectively. For $q$ odd, we further simplify the condition for $d$ to be 4, which is given by $\gcd(2h+1, q+1) = 1$. Combine these conditions and the dimension of $\mathcal{C}_{(q,q+1,3,h)}$, infinite families of MDS codes with parameters $[q+1,q-2,4]$ and AMDS codes with parameters $[q+1,q-3,4]$ are provided. Moreover, we show that such AMDS codes are  $d$-optimal and $k$-optimal locally repairable codes when $q>4h$. By applying the theorems in \cref{sec3}, we present several classes of AMDS codes which are $d$-optimal and $k$-optimal LRCs. The parameters of its dual codes $\mathcal{C}_{(q,q+1,3,h)}^\perp$ are also considered.
	
	In fact, when $q$ is odd, we completely determine the parameters of the BCH code $\mathcal{C}_{(q,q+1,3,h)}$ in all cases. When $q$ is even, however, it is not easy to verify whether there exist $x, y, z, w\in U_{q+1}$ such that \cref{sec3.1 thm2 equ0} holds. According to \cref{sec3 rmk2}, we know that $\mathcal{C}_{(q,q+1,3,h)}$ can potentially be an MDS code with parameters $[q+1,q-3,5]$ solely when $q$ is even. Therefore, when $\gcd(2h+1, q+1)=1$ and $h\notin\{0,\frac{q}{2},q\}$, it is worthwhile to ascertain whether the minimum distance of $\mathcal{C}_{(q,q+1,3,h)}$ is 4 or 5, or equivalently, to check whether the code $\mathcal{C}_{(q,q+1,3,h)}$ is AMDS or MDS. In addition, \cref{sec3 thm4} only provides a range for the minimum distance $d^\perp$ of its dual code $\mathcal{C}_{(q,q+1,3,h)}^\perp$. If this range can be narrowed down, or more precisely, $d^\perp$ can be determined, it will be very useful.


\end{document}